\documentclass[11pt]{article}
\usepackage{amssymb}
\usepackage{amsthm}
\usepackage{amsmath}
\newtheorem{theorem}{Theorem}[section]
\newtheorem{corollary}[theorem]{Corollary}

\usepackage{tikz,tikz-3dplot}
\usetikzlibrary{shapes.geometric}
\tikzstyle{scorestars}=[star, star points=5, star point ratio=2.25, draw,inner sep=1pt,anchor=center]
\usetikzlibrary{calc}
\usepackage{bbm}

\DeclareMathOperator*{\argmax}{arg\,max}
\usepackage{mathtools}
\definecolor{COLOR1}{HTML}{e6194b}
\definecolor{COLOR2}{HTML}{0082c8}
\definecolor{COLOR3}{HTML}{f58231}
\definecolor{COLOR4}{HTML}{911eb4}
\usetikzlibrary{patterns}
\usepackage{pgfplots}
\usepackage[justification=centering]{subcaption}
\usepackage[margin = 1in]{geometry}
\usepackage{amsmath}
\usepackage{algorithm}
\usepackage[noend]{algpseudocode}

\makeatletter
\def\BState{\State\hskip-\ALG@thistlm}
\makeatother

\newcommand{\motion}{motion}
\usepackage{stfloats}
\usepackage{csvsimple}
\usepackage{tabularx}
\newcolumntype{C}[1]{>{\centering\arraybackslash}m{#1}}

\begin{document}
\title{\vspace{-.7in}Multi-Agent Learning in Network Zero-Sum Games is a Hamiltonian System}  
	\author{%
	{James P. Bailey} \phantom{and} {Georgios Piliouras}\\
	\normalsize Singapore University of Technology and Design\\
	\normalsize $\{$james\_bailey,georgios$\}$@sutd.edu.sg
}
\date{}

\maketitle
\vspace{-.2in}
\begin{abstract}  
	Zero-sum games are natural, if informal, analogues of closed physical systems where no energy/utility can enter or exit. This analogy can be extended even further if we consider zero-sum network (polymatrix) games where multiple agents interact in a closed economy. 
	Typically, (network) zero-sum games are studied from the perspective of Nash equilibria. 
	Nevertheless, this comes in contrast with the way we typically think about closed physical systems, e.g., Earth-moon systems which move perpetually along recurrent trajectories of constant energy.

	We establish a formal and robust connection between multi-agent systems and Hamiltonian dynamics -- the same dynamics that describe conservative systems  in physics. Specifically, we show that no matter the size, or network structure of such closed economies, even if agents use different online learning dynamics from the standard class of Follow-the-Regularized-Leader, they yield Hamiltonian dynamics. 
	This approach generalizes the known connection to Hamiltonians for the special case of replicator dynamics in two agent zero-sum games developed by Hofbauer \cite{Hofbauer96}. 
	Moreover, our results extend beyond zero-sum settings and provide a type of a Rosetta stone (see e.g. Table \ref{tab:physics}) that  helps to  translate results and techniques  between online optimization, convex analysis,  games theory, and physics.  
\end{abstract}

\vspace{-.5in}
\begin{center}
	\begin{table*}[!b]
		\centering
		\begin{tabular}{|C{2.90in}  |C{2.90in}|}
			\hline
			\textbf{Mass on a Spring} & \textbf{Matching Pennies Gradient Descent}\\
			\hline
			\begin{tikzpicture}
			\draw[decoration={aspect=0.3, segment length=3mm, amplitude=5mm,coil},decorate] (0,0) -- (2.5,0); 
			\fill[white] (2.9,0) circle(.4);
			\shade[ball color = gray!40, opacity = 0.4] (2.9,0) circle (.4);
			\draw (2.9,0) circle (.4);
			\fill[gray, opacity=.4] (0,-.8)--(0,.8)--(-.1,.8)--(-.1,-.8)--cycle;
			\draw (0,-.8)--(0,.8)--(-.1,.8)--(-.1,-.8)--cycle;
			\node at (0,1.05) {};
			\end{tikzpicture}
			&$\left(\begin{array}{c c}x_1&1-x_1\end{array}\right)\left(\begin{array}{r r}1&-1\\-1&1\end{array}\right)\left(\begin{array}{c}x_2\\1-x_2\end{array}\right)$
			\\
			\hline
			Position $q(t)$ of Mass   & Agent 1 Strategy $x_1(t)$\\
			
			\begin{tikzpicture}[scale=.745]
			\draw (0,1)--(0,0);
			\draw (0,.5)--(4.5,.5);
			\fill (4.5,.5)--(4.3,.4)--(4.3,.6);
			\node[below] at (4.7,.45) {\footnotesize Time $t$};
			\csvreader[ head to column names,%
			late after head=\xdef\aold{\a}\xdef\bold{\b},%
			after line=\xdef\aold{\a}\xdef\bold{\b}]%
			{Position.dat}{}{%
				\draw (\aold, \bold) -- (\a,\b);
			}
			\end{tikzpicture}& 
			\begin{tikzpicture}[scale=.745]
			\csvreader[ head to column names,%
			late after head=\xdef\aold{\a}\xdef\bold{\b},%
			after line=\xdef\aold{\a}\xdef\bold{\b}]%
			{Position.dat}{}{%
				\draw (\aold, \bold) -- (\a,\b);
			}
			\draw (0,1)--(0,0);
			\draw (0,0)--(4.5,0);
			\fill (4.5,0)--(4.3,-.1)--(4.3,.1);
			\node[above] at (4.7,-.05) {\footnotesize Time $t$};
			
			\draw[dashed] (0,.5)--(4.5,.5);
			\node[left] at (0,.5) {\small $x_1^*$};
			\end{tikzpicture} \\
			\hline
			Momentum $p(t)$ of Mass  & Agent 2 Strategy $x_2(t)$\\
			\begin{tikzpicture}[scale=.745]
			\csvreader[ head to column names,%
			late after head=\xdef\aold{\a}\xdef\bold{\b},%
			after line=\xdef\aold{\a}\xdef\bold{\b}]%
			{velocity.dat}{}{%
				\draw (\aold, \bold) -- (\a,\b);
			}
			\draw (0,1)--(0,0);
			\draw (0,.5)--(4.5,.5);
			\fill (4.5,.5)--(4.3,.4)--(4.3,.6);
			\node[below] at (4.7,.45) {\footnotesize Time $t$};
			\end{tikzpicture}& 
			\begin{tikzpicture}[scale=.745]
			\csvreader[ head to column names,%
			late after head=\xdef\aold{\a}\xdef\bold{\b},%
			after line=\xdef\aold{\a}\xdef\bold{\b}]%
			{velocity.dat}{}{%
				\draw (\aold, \bold) -- (\a,\b);
			}
			\draw (0,1)--(0,0);
			\draw (0,0)--(4.5,0);
			\fill (4.5,0)--(4.3,-.1)--(4.3,.1);
			\node[above] at (4.7,-.05) {\footnotesize Time $t$};
			
			\draw[dashed] (0,.5)--(4.5,.5);
			\node[left] at (0,.5) {\small $x_2^*$};
			\end{tikzpicture} \\
			\hline
			$q(t)$ vs $p(t)$ & $x_1(t)$ vs $x_2(t)$\\
			\begin{tikzpicture}[scale=.806]
			\draw[->,>=stealth',semithick] (-90:1.2) arc[radius=1.2, start angle=-90, end angle=-210];
			\draw[->,>=stealth',semithick] (150:1.2) arc[radius=1.2, start angle=150, end angle=30];
			\draw[->,>=stealth',semithick] (30:1.2) arc[radius=1.2, start angle=30, end angle=-90];
			\draw[] (0,-1.5)--(0,1.5);
			\draw[] (-1.5,0)--(1.5,0);
			
			\end{tikzpicture}& 
			\begin{tikzpicture}[scale=.806]
			\draw[->,>=stealth',semithick] (-90:1.2) arc[radius=1.2, start angle=-90, end angle=-210];
			\draw[->,>=stealth',semithick] (150:1.2) arc[radius=1.2, start angle=150, end angle=30];
			\draw[->,>=stealth',semithick] (30:1.2) arc[radius=1.2, start angle=30, end angle=-90];
			\draw[dashed] (0,-1.5)--(0,1.5);
			\draw[dashed] (-1.5,0)--(1.5,0);
			\node[right] at (0,1.5) {$x_1^*$};
			\node[right] at (1.5,0) {$x_2^*$};
			\node[left] at (-1.5,0) {\phantom{$x_2^*$}};
			\end{tikzpicture} \\
			\hline
			Conservation of Energy & Constant Distance to Nash Equilibrium $x^*$\\
			$E=\frac{1}{2}kq^2(t)+\frac{1}{2m}p^2(t)$ & $D=\frac{1}{2}\left(x_1(t)-x_1^*\right)^2+\frac{1}{2}\left(x_2(t)-x_2^*\right)^2$\\[6pt]
			\hline
		\end{tabular}
		\caption{The physics of Matching Pennies when updated with Gradient Descent.}\label{tab:physics}
	\end{table*}
\end{center}

\newpage
\section{Introduction}



Undoubtedly equilibrium is the main notion through which games are studied and understood.
Not only are Nash equilibria the main solution concept but typically approaches that venture towards understanding the behavior of adaptive agents in games still naturally gravitate to slight modifications of the same idea. For example,  an evolutionary stable strategy \cite{smith1988evolution} is a Nash equilibrium refinement that captures attracting limit points of evolutionary learning dynamics. In the other direction, coarse correlated equilibria are generalizations of Nash equilibria with the property that the time average of the behavior of no-regret dynamics converge to them in general games \cite{young2004strategic}. 

Both of these equilibrium concepts, as well as many other axiomatics refinements or generalizations to the notion of Nash equilibrium share the same unifying attribute. 
All of these solution concepts are points. When we think of a game, i.e., an interaction of self-interested agents, we typically and almost unconsciously argue about its behavior using static points. We wish to argue that this way of thinking is limiting in a fundamental way as well as propose an alternative way of thinking about these systems.

To consider the severity of these limitations, consider the following example: Suppose that you were challenged to describe another many body interaction, e.g., the behavior of the Earth-Moon system, or the behavior of two masses connected by an ideal spring or even the behavior of coupled pendulums, but there was a catch. You were forced to output a single point of the state space as your description for this system. Clearly, this requirement is too restrictive as any meaningful model must be allowed to express non-equilibrium behavior and do so in a quantitative fashion (e.g. what is the shape, energy of these orbits for different initial conditions?). Instead, the behavior of these physical systems is captured by dynamical systems, specifically, Hamiltonian dynamics.

Interestingly, when we study  multi-agent games and learning dynamics on them, we do not find this restriction such a stifling roadblock, at least not intuitively.  Thus we implicitly make the assumption that the systems that we encounter when studying learning in games are fundamentally different from any of the above. After all what does gradient descent or multiplicative weights  in Matching Pennies have to do with an Earth-moon system? 

 In the contrast to our error prone intuition, laborious formal inquiries into the limitations of Nash equilibria, the foremost solution concept, have revealed an increasingly pessimistic picture about their applicability, at least as a general purpose model for  general-sum games. Not only are Nash equilibria hard to compute
\cite{Daskalakis06thecomplexity,chen2009settling,mehta2014constant} but they are also hard to approximate
\cite{skopalik2008inapproximability,rubinstein2015inapproximability,rubinstein2016settling}. Even their communication complexity seems  prohibitive for practical considerations \cite{hart2010long,Babichenko:2017:CCA:3055399.3055407}. Finally, a plethora of simple and even small games exists where numerous learning dynamics do not converge but instead can lead to cycles and other recurrent behavior \cite{hart2003uncoupled,daskalakis10,paperics11,Balcan12,papadimitriou2016nash,GeorgiosSODA18}. 
 What are we to make of all these  failures?

\textit{Our key take home message} in this paper is that besides these carefully documented formal failures of the Nash equilibrium, there exists a much more insidious failure of intuition about the efficacy of equilibrium, \textit{any sort of equilibrium}, as a model even for the most classic game theoretic instances. Revisiting our question, what does gradient descent or multiplicative weights  in Matching Pennies have to do with a mass system on spring? Or more generally, what is the connection between a $n$-body problem (e.g. our solar system) and a closed economy with numerous (e.g. thousands or millions) agents each possibly using different online learning algorithms from a large and diverse family of dynamics such as Follow-the-Regularized-Leader/Mirror-Descent \cite{Cesa06,hazan2016introduction}?  
 Surprisingly, we show that they are in a specific sense instantiations of the same phenomenon!  \textit{They are both  Hamiltonian systems} with notions of Hamiltonian energy that is preserved over time and which dictates the shape of the system trajectories (see Table \ref{tab:physics}). Critically, this is not merely an analogy but we present formal reductions from game theoretic dynamics to Hamiltonian systems. This opens new ways of thinking and arguing about games that circumvent altogether the need of using equilibria as a stepping stone. 
Despite Hamiltonians being an indispensable primitive in the study of most physical systems, this formal connection seems to have received little attention with few notable exceptions for special cases of this phenomenon, particularly for the case of replicator dynamics in zero-sum (as well as coordination) games \cite{Hofbauer96,Sato02042002,2018arXiv180205642B} and fictitious play \cite{Ostrovski2011,Strien2011}. We unify and generalize this understanding by exploiting connections between online learning and convex analysis. Along the way, we provide a natural interpretation about the emergence of recurrent behavior in zero-sum games for numerous online learning dynamics \cite{piliouras2014optimization,PiliourasAAMAS2014,GeorgiosSODA18,daskalakis2017training,nagarajan2018three} and divergence from Nash equilibria in discrete-time settings \cite{BaileyEC18}.  We briefly summarize our results below.




\textbf{Our contribution.}
We analyze zero-sum and coordination two agent normal-form games.  
We show that the dynamics created by applying the Follow-the-Regularized-Leader (FTRL) algorithm are Hamiltonian (Theorem \ref{thm:TwoAgent}). 
This reduction establishes the equivalence between a mass on a spring and gradient descent applied to Matching Pennies as shown in Table \ref{tab:physics}. 

Moreover, we extend these results to all network zero-sum games \cite{Cai,cai2016zero} and all coordination bipartite network games (Theorem \ref{thm:Idenity} and Corollary \ref{cor:bipartite}) demonstrating that the theory of Hamiltonian dynamics extends to an expansive set of games.
These results can even be extended to a larger class of payoff functions (Theorems \ref{thm:GeneralGame} and \ref{thm:GeneralBipartite}) which includes almost every two-agent, two-strategy, normal-form game (Corollary \ref{cor:2x2}). 
In fact, Corollary \ref{cor:2x2} only excludes trivial games where one of the agent's payout is independent of the other agent's actions.

Finally, we import some of the tools from the physics literature on Hamiltonians to unify several recent results on online learning in games. 
We show that the strategies in a network zero-sum  game stay approximately equidistant from the Nash equilibrium when updated by FTRL (Theorem \ref{thm:distance}). 
This generalizes a result from \cite{piliouras2014optimization} that was for replicator dynamics, a specific variant of FTRL.
This result follows from the ``invariant-energy'' property of Hamiltonian systems. 
Another common property of Hamiltonians is Poincar\'{e} recurrence which is formally shown for network zero-sum  games in \cite{piliouras2014optimization,GeorgiosSODA18}.
Finally, we revisit a recent result in  \cite{BaileyEC18} where it was shown that the discrete-time version of FTRL  in network zero-sum games repels strategies from the Nash equilibrium. This result follows intuitively from the fact that these discrete-time updates move along the tangent of the boundary of the  convex sublevel sets of the Hamiltonian/energy. Thus, these trajectories have non-decreasing energy (Theorem \ref{thm:increaseenergy}) and as a result move outwards towards the boundary (Corollary \ref{cor:diverge}).



\section{Preliminaries}\label{sec:Prelim}

A network game (or graphical polymatrix) consists of a set of $n$ agents ${\cal N}=\{1,2,...,n\}$ where agent $i$ has a finite set of actions ${\cal S}_i$ to choose from. 
Agent $i$ may select a \emph{mixed strategy}, $x_i$, from the $|{\cal S}_i|$-simplex ${\cal X}_i=\{x_i\in \mathbb{R}^{|{\cal S}_i|}_{\geq 0} : \sum_{s_i\in {\cal S}_i} x_{is_i}=1\}$. 
For each pair $\{i,j\}\subseteq {\cal N}$, agent $i$ has the payoff matrix $A^{(ij)}$ such that agent $i$ receives utility $x_i^\intercal A^{(ij)} x_j$ given agent $i$ and $j$'s strategies $x_i$ and $x_j$. 
This results in $n$ optimization problems where agents act strategically and independently to maximize their payouts.
\begin{equation}
\label{eqn:GeneralGame}\tag{Network Game}
\max_{x_i\in{\cal X}_i}\  x_i^\intercal \sum_{j\neq i}A^{(ij)} x_j \ \forall i\in {\cal N}
\end{equation}

A solution to a network game is a Nash equilibrium $x^*$. 
It satisfies 
\begin{align}
	x_i^* \sum_{j\neq i}A^{(ij)} x_j^* \geq x_i \sum_{j\neq i}A^{(ij)} x_j^* \ \forall x_i\in {\cal X}_i \ \forall i\in {\cal N}.\label{eqn:NashProp}
\end{align} 
Moreover, if $x_i^*$ is in the relative interior of ${\cal X}_i$ for each $i$ then it is \emph{fully-mixed} and (\ref{eqn:NashProp}) holds with equality.

A separable zero-sum multi-agent
game (zero-sum graphical game)~\cite{Cai}
is a graphical polymatrix (network) game in which 
the sum of all agent payoffs is always zero ($\sum_{i\in {\cal N}}x_i\sum_{j\neq i}A^{(ij)}x_j=0 \ \forall x\in \bigtimes {\cal X}_i$).
These are games that encode closed systems, an economy where no resources/utility can enter or exit and instead there is only exchange of utility between the agents. Although in this game formulation the edge games are not necessarily zero-sum there exists~\cite{Cai} a  (polynomial-time computable) payoff preserving transformation from every separable zero-sum multi-agent game to a pairwise constant-sum polymatrix game
 (\textit{i.e.}, a graphical polymatrix (network) game such that for each pair of agents $i,j:$
$A^{(ji)}=c_{\{j,i\}}\textbf{1}-\big(A^{(ij)}\big)^\mathrm{T}$ and $\textbf{1}$ the all-one matrix). 
 Removing the constants $c_{\{j,i\}}\textbf{1}$ from each edge game does not affect the strategic nature of the game nor the learning dynamics that we will consider.
 Hence for the rest of the paper we will focus on  games in the form $A^{(ij)}=-\left(A^{(ji)}\right)^\intercal$, i.e., polymatrix (network) zero-sum games \cite{dask09}.
 On the antipode of zero-sum games are \emph{coordination/partnership games} where at each outcome both agents get the same utility. 
  We can represent both of them using the notation $A^{(ij)}={\sigma}\big(A^{(ij)}\big)^\mathrm{T}$ where $\sigma=1$ for coordination games and $\sigma=-1$ for zero-sum games.
 

\subsection{Follow-the-Regularized-Leader (FTRL)}

In many settings, agents do not know the payoff matrices nor the Nash equilibrium. 
In such settings, agents adaptively update their strategies. 
The most well known class of algorithms for  online learning and optimization is Follow-the-Regularized-Leader. 
Given initial payoff vector $y_i(0)$, agent $i$ in (\ref{eqn:GeneralGame}) updates their strategies at time $t$ according to
\begin{equation}
\tag{FTRL}\label{eqn:FTRL}
\begin{aligned}
y_i(t)&= y_{i}(0)+\int_0^t \sum_{j\neq i}A^{(ij)}x_{j}(s) ds \\
x_i(t)&= \argmax_{x_i\in {\cal X}_i} \{\langle x_i, y_i(t)\rangle-h_i(x_i)\} 
\end{aligned}
\end{equation}
where $h_i$ is strongly convex and continuously differentiable. 
The cumulative payoff vector $y_i(t)$ indicates the cumulative payouts until time $t$ -- if agent $i$ had played strategy $s_i$ constantly since time $t=0$, then agent $i$ would receive a cumulative payout of $y_{is_i}(t)$. 

Strong convexity of $h_i$ and convexity and compactness of ${\cal X}_i$ guarantee (\ref{eqn:FTRL}) has a unique solution and therefore $x_i(t)$ is well-defined.
It is also well known that $x_i(t)=\nabla h_i^*(y_i(t))$ \cite{Shalev2012} where 
\begin{align}
h_i^*(y_i)=\max_{x_i\in {\cal X}_i} \{\langle x_i, y_i\rangle-h_i(x_i)\}
\end{align}
is the convex conjugate of $h_i$. 
The two most well-known versions of (\ref{eqn:FTRL}) are the gradient descent algorithm with $h_i(x_i)=||x_i||_2^2$ and the multiplicative weights algorithm, equivalently the replicator dynamics, with $h_i(x_i)=\sum_{s_i\in {\cal S}_i} x_{is_i}\log x_{is_i}$. 

\subsection{Hamiltonian dynamics}\label{sec:HamPhysics}

A physical system may consist of several interacting, moving objects (sometimes referred to as bodies or particles). 
To study each object, we track its position and momentum/velocity at time $t$, $q(t)$ and $p(t)$ respectively.  For simplicity we will assume that its mass is equal to one ($m=1$) so that momentum and velocity are equal.
In many systems, there exist natural laws preserving a relationship between the $q(t)$ and $p(t)$. 
Two such systems include a mass attached to a spring and a pendulum.

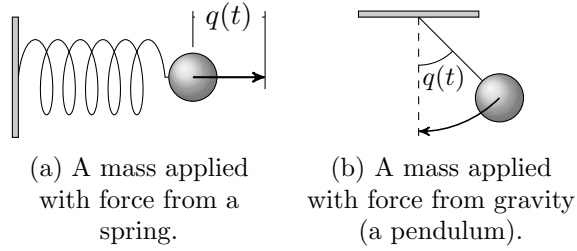
\begin{figure}[h]
	\begin{center}
		\begin{subfigure}[t]{0.22\textwidth}\centering
			\begin{tikzpicture}[scale=.8]
			\draw[decoration={aspect=0.3, segment length=3mm, amplitude=5mm,coil},decorate] (0,0) -- (2.5,0); 
			\fill[white] (2.9,0) circle(.4);
			\shade[ball color = gray!40, opacity = 0.4] (2.9,0) circle (.4);
			\draw (2.9,0) circle (.4);
			\fill[gray, opacity=.4] (0,-1)--(0,1)--(-.1,1)--(-.1,-1)--cycle;
			\draw (0,-1)--(0,1)--(-.1,1)--(-.1,-1)--cycle;
			\draw[->, >=stealth',thick] (2.9,0)--(4.1,0);
			\draw (4.1,-.2)--(4.1, 1.1);
			\draw (2.9,.5)--(2.9, 1.1);
			\draw (2.9,1)--(4.1, 1) node[midway, fill=white] {$q(t)$};
			
			\end{tikzpicture}
			\caption[]{A mass applied with force from a spring.}
		\end{subfigure}\hspace{.1in}
		\begin{subfigure}[t]{0.22\textwidth}\centering
			\begin{tikzpicture}[scale=.8]
			\draw (0,0)--(1.3435,-1.3435);
			
			\fill[gray, opacity=.4] (1,0)--(-1,0)--(-1,.1)--(1,.1)--cycle;	
			\draw (1,0)--(-1,0)--(-1,.1)--(1,.1)--cycle;	
			
			\fill[white] (1.3435,-1.3435) circle(.4);
			\shade[ball color = gray!40, opacity = 0.4] (1.3435,-1.3435) circle (.4);
			\draw (1.3435,-1.3435) circle (.4);
			
			\draw[->,>=stealth',semithick] (-45:1.9) arc[radius=1.9, start angle=-45, end angle=-90];	
			\draw[] (-45:.8) arc[radius=.8, start angle=-45, end angle=-90] ;
			\node at (-67.5:1.1) {\small $q(t)$};	
			\draw[dashed] (0,0)--(0,-2);
			
			\end{tikzpicture}
			\caption[]{A mass applied with force from gravity (a pendulum).}
		\end{subfigure}
		\caption[]{Physical systems of motion.}\label{fig:1} 
	\end{center}
\end{figure}

A well known concept from physics is ``conservation of energy''.  
If no outside force acts on the system, then energy is invariant at all times. Given a spring with constant $k$, a body with mass $m=1$, conservation of energy implies 
\begin{align}
\frac{1}{2}kq^2(0)+\frac{1}{2}p^2(0)=\frac{1}{2}kq^2(t)+\frac{1}{2}p^2(t)\tag{Energy of a Spring}\label{eqn:spring}
\end{align}
where $\frac{1}{2}kq^2(t)$ denotes the potential energy (the energy due to position) at time $t$ and $\frac{1}{2}p^2(t)$ denotes the kinetic energy (the energy due to motion) at time $t$.  

By definition, velocity is related to position with $p=\dot{q}$.
However, in some systems it is simpler to track the body's motion with something other than $\dot{q}$. 
For instance, a common way to track the pendulum is with the velocity of the body, $p(t)$, and to define the position, $q(t)$, as the angle the string forms with respect to the normal direction. 
Since $q(t)\in \mathbb{R}$  and $p(t)\in \mathbb{R}^2$, $p\neq \dot q$.
As such, we refer to $p(t)$ more generally as the \emph{motion} of the system.

Hamiltonian dynamics are an important subset of physical systems that includes systems as small as a mass on a spring or a pendulum and as large as the planetary orbits. 
Every Hamiltonian dynamic has a Hamiltonian $H(q,p)$ such that
\begin{equation}
\tag{Hamilton's equations}\label{eqn:HDyanamic}
\begin{aligned}
	\frac{\partial H}{\partial p}&= \frac{dq}{dt}=\dot{q}\\
	-\frac{\partial H}{\partial q}&= \frac{dp}{dt}=\dot{p}
\end{aligned}
\end{equation}
Hamiltonians have a variety of well known properties including conservation of energy (time invariance of the Hamiltonian) and volume preservation.  
There are even large families of integrators especially well suited for approximating solutions to system of differential equations given by (\ref{eqn:HDyanamic}).  
We discuss these properties  later in Section \ref{sec:Significance} as they relate to network games.

\section{The Physics of Two-Player Games}
{\color{black}
We begin our analysis with two player games. 
The position we will  track in FTRL is agent 1's cumulative strategy  $X_1(t)=\int_0^t x_1(s)ds$. 
Motion drives the change in position.  
As defined in (\ref{eqn:FTRL}), $y_1^t$ drives the change in $X_1^t$.  
Formally, $\frac{d}{dt} X_1(t)=x_1(t)=\nabla h^*_1(y_1(t))$ and it is natural to define $y_1(t)$ as the {\motion} of the system.

Observe that $y_2(t)=y_2(0)+\int_{0}^t A^{(21)}x_1(s)ds= y_2(0)+A^{(21)}X_1(t)$.  
Therefore tracking $y_1(t)$ and $X_1(t)$ gives us information about both agents via the perspective of agent 1.  
Symmetrically, position and motion can be defined from the perspective of agent 2 with $y_2(t)$ and $X_2(t)$ instead.
The idea of defining the system from the perspective a single agent comes naturally when we consider the dynamics in Table \ref{tab:physics}. 
The position of the mass on the spring corresponds to the agent 1's strategy and we associate the mass with agent 1.
The velocity of the mass corresponds to agent 2's strategy.
However, by associating the mass with agent 1, we have already framed our understanding of the system through agent 1's perspective. 

Given that agent strategies are often the object of interest when discussing the dynamics of games, it may be surprising that we are not tracking the actual strategy $x_i(t)$.  
However, the mapping $\nabla h_i^*: y_i \to x_i$ is surjective but not injective. 
Therefore we only gain information by tracking $y_i$ instead of $x_i$ suggesting that we can better understand the system using $y_i$.

Now that we have an understanding of the position and motion of the system, we can define the potential and kinetic energies. 
Just as gravity pulls an object toward the center of the Earth, the regularizer $h_2$ pulls $x_2(t)$ to the minimizer of $h_2$.  
Similarly, $h_2^*$ pulls $y_2(t)=y_2(0)+A^{(21)}X_1(t)$ to the minimizer of $h_2^*$ and we use $-\sigma h_2^*\left(y_2(0)+A^{(21)}X_1(t)\right)$ to represent the potential energy in the system. 
Analogously, the energy due to motion, i.e., kinetic energy, is  $h^*_1(y_1(t))$.}
Thus, the total energy in the system at time $t$ is 
\begin{align}
	H(X_1,y_1)=h_1^*(y_1(t))-\sigma h_2^*\left(y_2(0)+A^{(21)}X_1(t)\right)
\end{align}

\begin{theorem}\label{thm:2energy}
	The total energy $H(X_1,y_1)$ is invariant when a two-agent coordination or zero-sum (\ref{eqn:GeneralGame}) is updated with (\ref{eqn:FTRL}). 
\end{theorem}
\begin{proof} The result follows by taking the derivative of $H$ with respect to $t$.
	\begin{align}
		\frac{dH}{dt}&= \nabla h_1^*(y_1(t))^\intercal \frac{d}{dt}y_1(t) - \sigma \nabla h^*_2\left(y_2(0)+A^{(21)}X_1(t)\right)^\intercal A^{(21)}x_1(t)\\
		&= x_1(t)^\intercal A^{(12)}x_2(t)-  \nabla h^*_2(y_2(t))^\intercal \left(A^{(12)}\right)^\intercal x_1(t)\\ 
		&= x_1(t)^\intercal A^{(12)}x_2(t)-  x_2(t)^\intercal \left(A^{(12)}\right)^\intercal x_1(t)\\ 
		&=0		
	\end{align}
completing the proof of the theorem.	
	\end{proof}

\subsection{Hamiltonian dynamics of two-agent games}

Recall the mass attached to the spring from Section \ref{sec:HamPhysics}. 
Thus far, we have only discussed the position and motion of the mass as they relate to the total energy in the system and not the dynamical system created by the spring. 
However, in this case, the dynamical system is described entirely by the energy in the system.
Since the mass is $m=1$, the energy of the spring is invariant and given by
\begin{align}
	H(q,p)=\frac{1}{2}kq^2(t)+ \frac{1}{2}p^2(t).
\end{align}
Recall that  $F=\dot{p}$ where $F$ is force and $\dot{p}$ is acceleration (time derivative of momentum). 
The force enacted on a spring given position $q(t)$ is $-kq(t)$. Therefore 
\begin{align}
	\dot{p}= F= -{k}q(t)=-\frac{\partial H}{\partial q}.
\end{align}

Similarly, 
\begin{align}
	\dot{q}&= p(t)=\frac{\partial H}{\partial p}.
\end{align}

Thus the system is Hamiltonian.  Similarly, (\ref{eqn:FTRL}) is a Hamiltonian system.

\begin{theorem}\label{thm:TwoAgent}
	The dynamics (\ref{eqn:FTRL}) on a two-agent coordination or zero-sum (\ref{eqn:GeneralGame}) are Hamiltonian with invariant $H(X_1,y_1)$. 
\end{theorem}

\begin{proof}The proof simply follows by taking the partial derivative of $H(X_1,y_1)$ with respect to $y_1$ and $X_1$. 
\begin{align}
	\frac{\partial H}{\partial y_1}&=\nabla h^*(y_1(t))= x_1(t)=\frac{d}{dt} X_1(t)\\
	-\frac{\partial H}{\partial X_1}&=\sigma \left( A^{(21)}\right)^\intercal \nabla h^*_2\left(y_2(0)+A^{(21)}X_2(t)\right)\\
	&= A^{(12)} \nabla h^*_2(y_2(t))\\
	&= A^{(12)} x_2^t= \frac{d}{dt} y_1(t)
\end{align}
$H(X_1,y_1)$, $X_1$, and $y_1$ satisfy (\ref{eqn:HDyanamic}).  Therefore, the system is Hamiltonian with invariant energy $H(X_1,y_1)$. 
\end{proof}

\section{The Physics of Network Games}

In this section, we extend Theorems \ref{thm:2energy} and \ref{thm:TwoAgent} to network games with an arbitrary number of agents. 
First we extend our results to coordination and zero-sum bipartite network games. 
In a bipartite network game, agents are partitioned into disjoint ${\cal N}_1$ and ${\cal N}_2$. 
For agents $\{i,j\}\subseteq {\cal N}_k$, $A^{(ij)}$ is uniformly zero, i.e., there is no interaction between agents in the same partition. 
Second, we extend Theorem \ref{thm:TwoAgent} to all zero-sum versions of (\ref{eqn:GeneralGame}). 

\subsection{Coordination and Zero-Sum Bipartite Network Games}\label{sec:bipartite}

To extend the results to bipartite network games, we study the system from the perspectives of all agents in ${\cal N}_1$ via the variables $y_i(t)$ and $X_i(t)$ for $i\in {\cal N}_1$. 
We define the energy in the system as
\begin{align}
	H(X,y)= \sum_{i \in {\cal N}_1} h^*_i(y_i(t))-\sigma \sum_{j\in {\cal N}_2} h^*_j\left(y_j(0)+\sum_{i\in {\cal N}_1} A^{(ji)}X_i(t) \right).
\end{align}

With these definitions, we can immediately extend Theorem \ref{thm:TwoAgent} to bipartite network games.

\begin{corollary}\label{cor:bipartiteinvariance}
	$H(X,y)$ is time-invariant when (\ref{eqn:FTRL}) is applied to a bipartite coordination or zero-sum (\ref{eqn:GeneralGame}).
\end{corollary}

\begin{corollary}\label{cor:bipartite}
	The dynamics (\ref{eqn:FTRL}) on a bipartite coordination or zero-sum (\ref{eqn:GeneralGame}) are Hamiltonian with invariant $H(X,y)$. 
\end{corollary}

We prove both corollaries by reducing a bipartite network game to a two-agent game.

\begin{proof}
By generalizing the idea of a strategy space, a bipartite game can be expressed as a two-agent game. 
A bipartite game is written as
\begin{align}
	\max_{x_i\in {\cal X}_i} \left\{ x_i^\intercal \sum_{j\in {\cal N}_2} A^{(ij)} x_j  \right\} \ \forall i \in {\cal N}_1\\
	\max_{x_j\in {\cal X}_j} \left\{ x_j^\intercal \sum_{j\in {\cal N}_1} A^{(ji)} x_i  \right\} \ \forall j \in {\cal N}_2
\end{align}
Since there is no interaction between agents $i$ and $j$ when $\{i,j\}\subseteq N_k$, their optimization problem can be rewritten as 
\begin{align}
	\max_{\{x_i\}_{i\in {\cal N}_1}\in \bigtimes_{i\in {\cal N}_1}{\cal X}_i} \left\{ \sum_{i\in {\cal N}_1}x_i^\intercal \sum_{j\in {\cal N}_2} A^{(ij)} x_j  \right\}\label{eqn:3} \\
	\max_{\{x_j\}_{j\in {\cal N}_2}\in \bigtimes_{j\in {\cal N}_2}{\cal X}_j} \left\{ \sum_{j\in {\cal N}_2}x_j^\intercal \sum_{i\in {\cal N}_1} A^{(ji)} x_i  \right\} \label{eqn:4}
\end{align}
which is a two-agent game with the more general strategy space $\bigtimes_{i\in {\cal N}_k}{\cal X}_i$. 
Informally, this corresponds to two meta-agents $1$ and $2$ that simultaneously make decisions for agents in ${\cal N}_1$ and ${\cal N}_2$ respectively. 
Moreover, (\ref{eqn:FTRL}) yields the same dynamics on both formulations since the corresponding optimization problems are separable. 

Corollaries \ref{cor:bipartiteinvariance} and \ref{cor:bipartite} now follow directly from Theorems \ref{thm:2energy} and \ref{thm:TwoAgent}.  
First, observe that the proofs of Theorems \ref{thm:2energy} and \ref{thm:TwoAgent} only makes use of the fact that the strategy space is convex and compact to obtain the equality $x_i(t)=\nabla h^*_i(y_i(t))$.  
Thus, both theorems extend to games with arbitrary convex, compact strategy spaces and therefore to the bipartite network game described in (\ref{eqn:3}-\ref{eqn:4}).
\end{proof}

\subsection{Zero-Sum Network Games}\label{sec:GeneralNetwork}

Unlike two-agent and bipartite network games, there is no clear side/perspective to study the dynamics from in general network games. 
As result, we simultaneously examine all perspectives via  $X_i$ and $y_i$ for all $i\in {\cal N}$.  
This yields the energy function
\begin{equation}
H(X,y)= \sum_{i\in {\cal N}} h^*_i\left( y_i(t)\right) -\sigma \sum_{j\in {\cal N}}h^*_j\left(y_j(0)+\sum_{i\neq j} A^{(ji)}X_i(t)       \right).
\end{equation}

\begin{theorem}\label{thm:energy}
	The total energy $H(X,y)$ is invariant when a zero-sum (\ref{eqn:GeneralGame}) is updated with (\ref{eqn:FTRL}). 
\end{theorem}

\begin{theorem}\label{thm:Idenity}
	The dynamics (\ref{eqn:FTRL}) on a zero-sum (\ref{eqn:GeneralGame}) are Hamiltonian with invariant $H(X,y)$. 
\end{theorem}

The proof of the theorems follow identically to Theorems \ref{thm:2energy} and \ref{thm:TwoAgent} respectively. 
Interestingly, the theorems also hold for coordination network games.  
However, the energy function $H(X,y)$ is meaningless in such settings. 
Recall that $y_i(t)=y_i(0)+\sum_{j\neq i} A^{(ij)}X_j(t)$. 
Thus, the energy function can be rewritten as
\begin{equation}
H(X,y)= (1-\sigma)\sum_{i\in {\cal N}} h^*_i\left( y_i(t)\right)=
\begin{cases}
2\sum_{i\in {\cal N}} h^*_i\left( y_i(t)\right) & (\sigma =-1)\\
0 &  (\sigma =1)\\
\end{cases}
\end{equation}
In the case of zero-sum network games ($\sigma=-1$), we are effectively double counting the energy in the system. 
However, in coordination network games ($\sigma=1$), the total energy is uniformly zero for every possible trajectory of (\ref{eqn:FTRL}) which provides no information about the actual dynamics.

\section{Using Tools from Physics to Gain Insight about Online Learning}\label{sec:Significance}

Typically, the study of games is centered around the Nash equilibrium. 
As such, it may be surprising that the Nash equilibrium never appears in our analysis of the behavior of learning dynamics in games. 
However, as we see in this section, our analysis has many implications regarding learning dynamics and the Nash equilibrium.
Recall, every solution to (\ref{eqn:GeneralGame}) is a Nash equilibrium $x^*_i\in {\cal X}_i$ and is such that 
\begin{equation}\label{eqn:NE}
x_i^*\sum_{i\neq j}A^{(ij)}x_j^* \geq x_i\sum_{i\neq j}A^{(ij)}x_j^* \ \forall x_i\in {\cal X}_i
\end{equation}
Moreover, if $x^*$ is in the relative interior of ${\cal X}$, then we refer to $x^*$ as fully-mixed and (\ref{eqn:NE}) holds with equality.

\subsection{Conservation of energy}\label{sec:EnergyIsDistnance}
It is well known that Hamiltonians preserve energy, providing an alternative proof of Theorem \ref{thm:2energy}. 
This conservation has significant implications for online learning in zero-sum games. 
In many settings, such as economics, the ``market'' is studied from the perspective of the Nash equilibrium.
However,  conservation of energy implies that agent strategies never get too close to a fully-mixed Nash equilibrium when agents are adaptive. 
Therefore adaptive agents never actually realize equilibrium behavior and may exhibit behaviors that are unforeseen to equilibrium analysis.

\begin{theorem}\label{thm:distance}
	Suppose $x^*$ is a fully-mixed Nash equilibrium of a zero-sum (\ref{eqn:GeneralGame}), 
	$x_i$ and $y_i$ are updated according to (\ref{eqn:FTRL}), and  $x(0)\neq x^*$.  There exists a $d>0$ such that the Bregman distance from $x^*$ to $x(t)$ is always at least $d$. Moreover, if the trajectory remains in the interior of the simplex for all times $t$ then the Bregman distance from $x^*$ to $x(t)$ is time-invariant.
\end{theorem}

\begin{proof} 
      We start by iterating on a slight variation of an argument made in \cite{GeorgiosSODA18} that the Fenchel-coupling (specifically in the sense of \cite{MS16}) between $x^*$ and $y(t)$ is time-invariant.
	The Fenchel-coupling is defined as
\begin{align}
	F(x^*,y(t))&\equiv\sum_{i\in{\cal N}} \left(h^*_i(y_i(t)) -\langle y_i(t),x_i^*\rangle+h_i(x_i^*)\right)
	\end{align}
	First, observe that $\sum_{i\in {\cal N}} \langle y_i(t),x_i^*\rangle$ is time-invariant:
	\begin{align}
	\sum_{i\in {\cal N}} \langle y_i(t),x_i^*\rangle&= \sum_{i\in {\cal N}}\langle y_i(0), x_i^*\rangle +\sum_{i\in {\cal N}} \int_0^t \left(x_i^*\right)^\intercal \sum_{j\neq i} A^{(ij)}x_j(s)ds\\
	&= \sum_{i\in {\cal N}}\langle y_i(0), x_i^*\rangle + \int_0^t \sum_{i\in {\cal N}}\sum_{j\neq i}\left(x_i^*\right)^\intercal  A^{(ij)}x_j(s)ds\\
	&= \sum_{i\in {\cal N}}\langle y_i(0), x_i^*\rangle + \int_0^t \sum_{j\in {\cal N}}\sum_{i\neq j}\left(x_i^*\right)^\intercal  A^{(ij)}x_j(s)ds\\
	&= \sum_{i\in {\cal N}}\langle y_i(0), x_i^*\rangle + \int_0^t \sum_{j\in {\cal N}}\sum_{i\neq j}\left(x_i^*\right)^\intercal  A^{(ij)}x_j^*\label{eqn:fullymixed}\\
	&= \sum_{i\in {\cal N}}\langle y_i(0), x_i^*\label{eqn:zerosum}\rangle
	\end{align}
	where (\ref{eqn:fullymixed}) follows since $x^*$ is fully-mixed and (\ref{eqn:zerosum}) follows since the game is zero-sum.
	By definition of $X_i(t)$ and $y_i(t)$, $H(X,y)=2  \sum_{i\in {\cal N}}h_i^*(y_i(t))$. 
	Therefore, the Fenchel-coupling is
	\begin{align}
		F(x^*,y(t))&=\sum_{i\in{\cal N}} \left(h^*_i(y_i(t)) -\langle y_i(t),x_i^*\rangle+h_i(x_i^*)\right)\\
	&=  \sum_{i\in{\cal N}} \left(h^*_i(y_i(t)) -\langle y_i(0),x_i^*\rangle+h_i(x_i^*)\right)\\
	&= \frac{1}{2}\cdot H(X,y)+\sum_{i\in{\cal N}}\left( h_i(x_i^*)-\langle y_i(0),x_i^*\rangle\right).
\end{align} 
Thus, the Fenchel-coupling and half of the energy in the system differ by a constant implying the Fenchel-coupling is time-invariant.
The Bregman distance from $x^*$ to $x(t)$ is defined as 
\begin{align}
	D(x^*||x)\equiv\sum_{i\in{\cal N}} \left( h_i(x_i^*)-h_i(x_i(t)) -\langle \nabla h_i(x_i(t)), x_i^*-x_i(t)\rangle \right).
\end{align}
The distance $D$ is uniquely minimized at $x=x^*$ with $D(x^*||x^*)=0$.
Moreover, since the boundary of ${\cal X}$ is compact and $D$ is continuous, there exists a $d>0$ such that $D(x^*||x)\geq d$ for all $x$ on the boundary. 
By \cite{MS16}, $F(x^*,y(t))\geq D(x^*||x(t))$ where equality holds whenever $x(t)$ is fully mixed.
Therefore, when $x(t)$ is fully mixed, $D(x^*||x(t))=F(x^*,y(t))=F(x^*,y(0))\geq D(x^*||x(0))>0$  since $F$ is time-invariant and $x(0)\neq x^*$.  
Thus, $D(x^*||x(t))\geq \min\{d, D(x^*||x(0))\}>0$.  
If 
  the trajectory remains in the interior of the simplex for all $t$
 then $D(x^*||x(t))=F(x^*,y(t))$ is constant for all $t$, completing the proof.
\end{proof}

\begin{corollary}[Piliouras and Shamma \cite{piliouras2014optimization}]\label{cor:KL}
	Suppose $x^*$ is a fully-mixed Nash equilibrium of a  zero-sum (\ref{eqn:GeneralGame}). Then the KL-divergence between $x^*$ and $x(t)$ is time-invariant when $x$ and $y$ are updated with the replicator dynamics.
\end{corollary}

\begin{proof}
Replicator dynamics and the KL-divergence are special cases of (\ref{eqn:FTRL}) and the Bregman distance respectively where $h(x_i)=x_i \log{x_i}$. 
Moreover, replicator dynamics guarantee that $x_i(t)$ is interior for all $t$  
and the result follows from Theorem \ref{thm:distance}.  
\end{proof}


\subsection{Recurrence in adversarial learning} 
 Poincar\'{e}  proved that in certain systems almost all trajectories return arbitrarily close to their initial position infinitely often, i.e., exhibit Poincar\'{e} recurrence  \cite{Poincare1890}.
Recurrent behavior is important in the applications of online learning because it implies that the behavior (approximately) repeats itself, becoming more predictable overtime.

Poincar\'{e} recurrence is typically shown via two properties: volume preservation and bounded orbits. 
\cite{piliouras2014optimization, GeorgiosSODA18} argue both of these properties directly to establish recurrence in the setting of network zero-sum games.
However, Hamiltonians provide the intuition for both of these arguments.
Volume preservation is a property held by all Hamiltonian systems. 
Moreover, many natural Hamiltonian systems have bounded orbits and therefore are Poincar\'{e} recurrent.  
In the case of zero-sum games, as we discussed in the previous section, the energy in the system corresponds to the total distance from the Nash equilibrium and, unsurprisingly, (\ref{eqn:FTRL}) on network zero-sum games with interior Nash has bounded orbits\footnote{In general, $y_i(t)$ drifts according to $\sum_{j\neq i}A^{(ij)}x^*_j$ and therefore often is not bounded.  However, adding a constant to the payoff matrix does not change the dynamics of $x(t)$ and therefore we may always assume $\sum_{j\neq i}A^{(ij)}x^*_j=0$ and $y_i(t)$ is bounded.} and therefore is Poincar\'{e} recurrent.

\subsection{Divergence from equilibrium in discrete-time}
The results in continuous-time make it clear that the Nash equilibrium is not always an appropriate tool to analyze a system where agents adaptively update their strategies. 
The study of Hamiltonians also shows thats the Nash equilibrium can be a poor solution concept in the discrete-time.
Specifically, in (\ref{eqn:FTRL}), agent strategies tend to be repelled from the Nash equilibrium in zero-sum games.

\begin{theorem}\label{thm:increaseenergy}
	Suppose a continuous dynamic $y(t)$ has an invariant energy $H(y)$.
	If $H$ is continuous with convex sublevel sets then the energy in the corresponding discrete-time dynamic obtained via Euler's method/integration is non-decreasing. 
\end{theorem}

\begin{proof}
         Let $t$ denote the current time instant.  
		Euler's method  with step-size $\eta$ yields an approximation of $y(t+\eta)$ with
	\begin{align}
	\hat{y}^{t+\eta}= y(t)+\eta \frac{d}{dt}y(t)
	\end{align}

	We will show $H(\hat{y}^{t+\eta})\geq H(y(t))$.
 	Suppose $H(y(t))=c$ and without loss of generality, assume $\{y: H(y)\leq c\}$ is full-dimensional. 
	Since $\{y: H(y)\leq c\}$ is convex, there exists a supporting hyperplane $\{y: a^\intercal y = a^\intercal y(t)\}$ such that $a^\intercal y \leq a^\intercal y(t)$ for all $y\in \{y: H(y)\leq c\}$.
	Therefore, 	
	\begin{align}
		a^\intercal \left(\frac{d}{dt}y(t)\right)&= a^\intercal \left(\lim_{s\to 0^+} \frac{y(t)-y(t-s)}{s}\right)\\
		&= \left(\lim_{s\to 0^+} \frac{a^\intercal y(t)-a^\intercal y(t-s)}{s}\right)\\
		&\geq \left(\lim_{s\to 0^+} \frac{a^\intercal y(t)-a^\intercal y(t)}{s}\right)=0,
	\end{align}
	implying
	\begin{align}
		a^\intercal \hat{y}^{t+\eta}&=a^\intercal y(t)+a^\intercal\left(\eta \frac{d}{dt}y(t)\right)\\	
									&\geq a^\intercal y(t).
	\end{align}
	
	For contradiction, suppose $H(\hat{y}^{t+\eta})< c$. By continuity of $H$, 
	 for sufficiently small $\epsilon>0$, $\hat{y}^{t+\eta}+\epsilon a \in \{y: H(y)\leq c\}$.  
	However, 
	\begin{align}
		a^\intercal (\hat{y}^{t+\eta}+\epsilon a) \geq a^\intercal y(t) +\epsilon||a||_2^2 > a^\intercal y(t)
	\end{align}
	contradicting that $\{y: a^\intercal y = a^\intercal y(t)\}$ is a supporting hyperplane. 
	Thus, the statement of the theorem holds.
\end{proof}

\begin{corollary}[Bailey and Piliouras \cite{BaileyEC18}]\label{cor:diverge}
	If $x^*$ is fully-mixed Nash equilibrium of a zero-sum (\ref{eqn:GeneralGame}), then the Fenchel-coupling between $y$ and $x^*$ is non-decreasing in the discrete-time version of (\ref{eqn:FTRL}). 
\end{corollary}

\begin{proof}
	The discrete-time version of (\ref{eqn:FTRL}) is simply Euler's method applied to the continuous time dynamics. 
	The convex-conjugate $h^*_i$ is convex by definition and therefore the energy function $H(y)= \sum_{i\in {\cal N}} h^*_i(y_i(t))$ is also convex and therefore continuous.  
	Hence, for $y,y'\in \{y: H(y)\leq c\}$, 
	\begin{align}H(\alpha y + (1-\alpha)y') \leq \alpha H(y)+(1-\alpha) H(y')\leq c.
	\end{align}
	Therefore $\alpha y + (1-\alpha)y'\in \{y: H(y)\leq c\}$ and $H$ has convex sublevel sets.
	By Theorem \ref{thm:Idenity}, the dynamic $y(t)$ has invariant energy $H(y)$.
	Therefore by Theorem \ref{thm:increaseenergy}, the energy in the discrete-time dynamics is non-decreasing.
	In the proof of Theorem \ref{thm:distance}, it was shown the Fenchel-coupling differs from the energy in the system by a constant when $x^*$ is fully-mixed.  
	Therefore, the Fenchel-coupling is non-decreasing.
\end{proof}

The main result of \cite{BaileyEC18} is that the strategies come arbitrarily close to the boundary infinitely often. 
This can also be shown using Hamiltonians;
Theorem \ref{thm:increaseenergy} becomes strictly increasing with strict convexity. 
Using continuity and compactness arguments, Corollary \ref{cor:diverge} can be strengthened to show that the Fenchel-coupling and Bregman distance strictly increases by a positive constant whenever $x(t)$ is in a compact region in the relative interior of ${\cal X}_i$. Since the Bregman distance is continuous, every compact region in the relative interior of ${\cal X}_i$ has a finite maximum Bregman distance and therefore $x(t)$ must leave the region infinitely often.  
This holds for every compact region in the relative interior of ${\cal X}_i$ and therefore $x(t)$ must come arbitrarily close to the boundary infinitely often.

\subsection{Intuitive, accessible explanations of learning dynamics}
Online learning is used in practice 
as a black box for adaptive, low regret decision making. 
 We have shown that Hamiltonians systems provide tools to significantly advance our understanding of such learning dynamics. 
Just as importantly, Hamiltonians provide the tools  to explain the results of learning dynamics to someone with no understanding of differential equations or analysis and only a basic understanding of physics.

Our analysis draws analogues between agent strategies, the NE, and the regularizer in (\ref{eqn:FTRL}) and the Earth, sun, and gravity respectively. 
In our solar system, the Earth's momentum, propels the Earth in a direction normal to the sun, and gravity pulls the Earth inwards towards the resulting in a stable periodic orbit. 
This equivalence perfectly explains the non-convergent, recurrent behavior of (\ref{eqn:FTRL}) in a (\ref{eqn:GeneralGame}). 
Further, a basic understanding of calculus provides intuition for discrete-time dynamics. 
The discrete-time dynamics are obtained simply by ``moving tangentially to the curve''. 
Since continuous time strategies orbit around the Nash equilibrium, discrete-time dynamics must diverge as in Corollary \ref{cor:diverge}.

\section{Generalizations of Games}

In this section, we extend the theory to a more general class of games. 
As we discussed in Section \ref{sec:bipartite}, our results hold for any convex, compact strategy space, ${\cal X}_i$.  
We further extend the results for the more general payoff function
\vspace{-.03in}\begin{equation}
\max_{x_i\in {\cal X}_i}\left\{\sum_{j\neq i}\left(x_i^\intercal A^{(ij)}x_j + b^{(ij)}\cdot x_i + d^{(ij)}\cdot x_j+c^{(ij)} \right)\right\} \ \forall i \label{eqn:general}\tag{Generalized Network Game}
\end{equation}
where $A^{(ij)}=\sigma \left( A^{(ji)}\right)^\intercal$. Again $\sigma=1$ corresponds to a coordination (\ref{eqn:general}) and $\sigma=-1$ corresponds to a zero-sum (\ref{eqn:general}).

With this new utility function, the dynamics of (\ref{eqn:FTRL}) must be redefined.  
Observe that agent $i$ only has control of $x_i^\intercal A^{(ij)}x_j + b^{(ij)}\cdot x_i$ and therefore only learns according to this component of their utility function. 
Thus, (\ref{eqn:FTRL}) is rewritten as
\begin{equation}\tag{FTRL2}\label{eqn:FTRL2}
\begin{aligned}
y_i(t)&= y_i(0)+\int_0^t \sum_{j\neq i}\left(A^{(ij)}x_j(s) + b^{(ij)}\right)ds\\
x_i(t)&= \argmax_{x_i\in {\cal X}_i} \{\langle x_i, y_i(t)\rangle-h_i(x_i)\}.
\end{aligned}
\end{equation}

Similar to $(\ref{eqn:FTRL})$, agent $i$ is still learning against agent $j$ via the payoff matrix $A^{(ij)}$.  
However, unlike a (\ref{eqn:GeneralGame}), a (\ref{eqn:general}) is not a closed system;
additional energy is introduced into the dynamical system via $b^{(ij)}$. 
Nonetheless, (\ref{eqn:FTRL2}) also admits Hamiltonian dynamics after describing motion with $y_i^t$ and position with $X_i(t)$.  
However, we modify the energy function to remove the additional energy introduced by $b^{(ij)}$:
\begin{equation}
\begin{aligned}
H(X,y)=& \sum_{i\in {\cal N}} h^*_i\left( y_i(t)\right) -\sigma \sum_{j\in {\cal N}}h^*_j\left(y_j(0)+\sum_{i\neq j}\left( A^{(ji)}X_i(t) +b^{(ji)}\right)\right)\\
&- \sum_{i\in {\cal N}}\sum_{j\neq j} b^{(ij)}\cdot X_i(t).
\end{aligned}
\end{equation}

\begin{theorem}\label{thm:GeneralGame}
	Given any set of convex, compact ${\cal X}_i$, the dynamics (\ref{eqn:FTRL2}) on a zero-sum (\ref{eqn:general}) are Hamiltonian with invariant $H(X,y)$. 
\end{theorem}

The proof follows identically to Theorem \ref{thm:Idenity}.

Similar to Theorem \ref{thm:TwoAgent} and Corollary \ref{cor:bipartite}, we can define the energy from one side of the network for coordination games in two-agent and bipartite (\ref{eqn:general}).
\begin{equation}
\begin{aligned}
	\bar{H}(X,y)=& \sum_{i\in {\cal N}_1} h^*_i\left( y_i(t)\right) -\sigma \sum_{j\in {\cal N}_2}h^*_j\left(y_j(0)+\sum_{i\in {\cal N}_i}\left( A^{(ji)}X_i(t)+b^{(ji)}\right) \right)\\
	&- \sum_{i\in {\cal N}_1}\sum_{j\in {\cal N}_2} b^{(ij)}\cdot X_i(t).
\end{aligned}
\end{equation}
\begin{theorem}\label{thm:GeneralBipartite}
	Given any set of convex, compact ${\cal X}_i$, the dynamics (\ref{eqn:FTRL2}) on a coordination or zero-sum, two-agent or bipartite (\ref{eqn:general}) are Hamiltonian with invariant $\bar{H}(X,y)$. 
\end{theorem}

It may not be immediately clear how the results on (\ref{eqn:general}) fits into the literature on learning theory.
Certainly, the results imply our results in the setting of (\ref{eqn:GeneralGame}) by taking $b^{(ij)}=d^{(ij)}=c^{(ij)}=0$ and ${\cal X}_i$ to be the standard simplex. 
However, there is a more nuanced reduction from (\ref{eqn:GeneralGame}) to (\ref{eqn:general}) that allows us to extend our results to a larger class of network games.

In (\ref{eqn:GeneralGame}), the strategy space is ${\cal X}_i=\{x_i\in \mathbb{R}^{|S_i|}_{\geq 0}: \sum_{s_i\in{\cal S}_i} x_{is_i}=1\}$.  
After selecting $s_i\in {\cal S}_i$ arbitrarily for each $i$, substituting $x_{is_i}=1-\sum_{s_i'\in{\cal S}_i\setminus\{s_i\}} x_{is'_i}$ into (\ref{eqn:GeneralGame}) yields
\begin{align}
	\max_{\bar{x}_i\in \bar{\cal X}_i}\left\{\sum_{j\neq i}\left( \bar{x}_i^\intercal \bar{A}^{(ij)} \bar{x}_j+b^{(ij)}\cdot \bar{x}_i +d^{(ij)}\cdot \bar{x}_j+c^{(ij)}\right)\right\} \ \forall i
\end{align}
where $\bar{\cal X}_i=\{\bar{x}_i\in \mathbb{R}^{|{\cal S}_i|-1}_{\geq 0}: \sum_{s'_i\in {\cal S}_i\setminus \{s_i\}} \bar{x}_{is'_i}\leq 1\}$. 
It is straightforward to verify that this substitution preserves the dynamics $\{x^t\}_{t=0}^\infty$ since the optimizer of a function does not change after a variable substitution. 
Under this formulation, the dynamics are Hamiltonian if $\bar{A}^{(ij)}=\sigma \left(\bar{A}^{(ji)}\right)^\intercal$ even if the condition does not hold for the original payoff matrix $A^{(ij)}$. 
As a result, almost every two agent, two strategy game admits Hamiltonian dynamics under (\ref{eqn:FTRL}).

\begin{corollary}\label{cor:2x2}
	Suppose ${\cal X}_i=\{x_i\in \mathbb{R}^2_{\geq 0}: x_{i1}+x_{i2}=1\}$. 
	Then almost every two-agent, two-strategy (\ref{eqn:GeneralGame}) can be expressed as a coordination or zero-sum version of (\ref{eqn:general}) and therefore has Hamiltonian dynamics. 
\end{corollary}

\begin{proof} Suppose the agents' payoff matrices are given by 
\begin{align}
	A^{(12)}= \left( \begin{array}{c c} a & b \\ c & d \end{array} \right)	\hspace{.2in}A^{(21)}= \left( \begin{array}{c c} \alpha & \beta \\ \gamma & \delta \end{array} \right).
\end{align}
After performing the substitution $x_{i2}=1-x_{i1}$, the two-agent game can be rewritten as 
\begin{align}
\max_{{x}_{11}\in [0,1]}& \{ (a+d-b-c){x}_{11}{x}_{21}+(b-d){x}_{11} + (c-d){x}_{21}+d  \} \label{eqn:1}\\
\max_{{x}_{21}\in [0,1]}& \{ (\alpha+\delta-\beta-\gamma){x}_{11}{x}_{21}+(\beta-\delta){x}_{21} + (\gamma-\delta){x}_{11}+\delta  \}.\label{eqn:2}
\end{align}

When performing this variable substitution, we must also update our definitions of $y_i(0)$ and $h_i$: 
\begin{align}
x_i(0)	&=\argmax_{x_i\in {\cal X}_i}\{ \langle x_i, y_i(0)\rangle -h_i(x_i)\}\\
		&=\argmax_{x_i\in {\cal X}_i}\{ x_{i1}\cdot y_{i1}(0)+x_{i2}\cdot y_{i2}(0) -h_i(x_{i1},x_{i2})\}\\
		&=\argmax_{x_i\in {\cal X}_i}\{ x_{i1}\cdot y_{i1}(0)+(1-x_{i1})\cdot y_{i2}(0) -h_i(x_{i1},1-x_{i1})\}\\
		&=\argmax_{x_i\in {\cal X}_i}\{ x_{i1}\cdot (y_{i1}(0)-y_{i2}(0)) -h_i(x_{i1},1-x_{i1})\}
\end{align}
Therefore the new regularizer is $\bar{h}_i(x_{i1})=h(x_{i1},1-x_{i1})$ and the new initial payoff vector is $\bar{y}_i(0)=y_{i1}(0)-y_{i2}(0)$.  
It is straightforward to verify that these substitutions yield the same dynamics $x_i(t)$. 
For simplicity, we write $x_{i1}$ as $x_i$, $\bar{y}_i$ as $y_i$, $\bar{h}_i$ as $h_i$.
Further, we rewrite (\ref{eqn:1}-\ref{eqn:2}) as
\begin{align}
\max_{{x}_1\in [0,1]}& \{a^1{x}_{1}{x}_{2}+b^1{x}_{1} + d^1{x}_{2}+c^1  \}\\
\max_{{x}_2\in [0,1]}& \{ a^2{x}_{1}{x}_{2}+b^2{x}_{2} + d^2{x}_{1}+c^2  \}.
\end{align}

For Theorem \ref{thm:GeneralGame} to apply to this game, we need for $a^1=\sigma a^2$ for some $\sigma \in \{\pm 1 \}$.
To accomplish this, we will further modify agent $2$'s payoff function, initial payoff vector $\bar{y}_2(0)$, and regularizer $\bar{h}_2$ in a way that preserves the dynamics of (\ref{eqn:FTRL2}). 
First, observe that for almost every two-agent, two-strategy game, $a^1=a+d-b-c\neq 0$ and $a^2=\alpha+\delta-\beta-\gamma\neq 0$.
Therefore, we can rewrite agent 2's portion of the game as
\begin{align}
&\max_{{x}_2\in [0,1]} \left\{ \frac{|a^1|}{|a^2|}a^2{x}_{1}{x}_{2}+\frac{|a^1|}{|a^2|}b^2{x}_{2} + \frac{|a^1|}{|a^2|}d^2{x}_{1}+\frac{|a^1|}{|a^2|}c^2  \right\}\\
=& \max_{{x}_2\in [0,1]} \left\{ \sigma a^1{x}_{1}{x}_{2}+\bar{b}^2{x}_{2} + \bar{d}^2{x}_{1}+\bar{c}^2  \right\}
\end{align}
where $\sigma\in \{\pm 1\}$.  
To preserve the dynamics $x_2(t)$, we must pass the change in the payoff function into the initial payoff vector $y_2(0)$ and regularizer $h_2$.  
This is accomplished with  $\bar{y}_2=\frac{|a^1|}{|a^2|}y_2$ and $\bar{h}_2=\frac{|a^1|}{|a^2|}h_2$. 
Once again, the dynamics $\{x^t\}_{t=0}^\infty$ are preserved since the substitutions rescale the maximization problem by the positive constant $\frac{|a^1|}{|a^2|}$.
Thus, after performing these substitutions, the dynamics are Hamiltonian by Theorem \ref{thm:GeneralGame}. 
\end{proof}

\section{Conclusions}

We have shown that online learning in games is a Hamiltonian dynamical system.
Moreover, this characterization  provides a deeper understanding of games played in both continuous and discrete time. 
In fact, several of the recent results on online learning can be described simply as a byproduct of the Hamiltonian behavior. 
Research on Hamiltonian dynamics is rich and encompasses far more than what is discussed in this paper.
What are the other consequences of Hamiltonian dynamics and what do they mean about online learning in games?

\section*{Acknowledgements}

James P. Bailey and Georgios Piliouras acknowledge SUTD grant SRG ESD 2015 097, MOE AcRF Tier 2 Grant 2016-T2-1-170,  grant PIE-SGP-AI-2018-01 and NRF 2018 Fellowship NRF-NRFF2018-07.

\bibliographystyle{acm}  
\bibliography{IEEEabrv,Bibliography,refer}

\end{document}